\documentclass[a4paper]{article}
\usepackage{amssymb}
\usepackage{amsfonts}
\usepackage{amsmath}
\usepackage{geometry}
\usepackage{graphicx}

\newtheorem{theorem}{Theorem}
\newtheorem{acknowledgement}[theorem]{Acknowledgement}

\newtheorem{assumption}[theorem]{Assumption}

\newtheorem{condition}[theorem]{Condition}

\newtheorem{proposition}[theorem]{Proposition}
\newtheorem{remark}[theorem]{Remark}

\newenvironment{proof}[1][Proof]{\noindent\textbf{#1.} }{\ \rule{0.5em}{0.5em}}
\input{tcilatex}
\begin{document}

\title{On a initial value problem arising in mechanics}
\author{Teodor M. Atanackovic%
\begin{footnote}
{Department of Mechanics, Faculty of Technical Sciences,
University of Novi Sad, Trg D. Obradovica 6, 21000 Novi Sad,
Serbia, atanackovic@uns.ac.rs}
\end{footnote}, Stevan Pilipovic%
\begin{footnote}
{Department of Mathematics, Faculty of Natural Sciences and
Mathematics, University of Novi Sad, Trg D. Obradovica 4, 21000
Novi Sad, Serbia, stevan.pilipovic@dmi.uns.ac.rs}
\end{footnote} and Dusan Zorica%
\begin{footnote}
{Mathematical Institute, Serbian Academy of Arts and Sciences,
Kneza Mihaila 36, 11000 Beograd, Serbia, dusan\textunderscore
zorica@mi.sanu.ac.rs}
\end{footnote}}
\maketitle

\begin{abstract}
We study initial value problem for a system consisting of an integer order
and distributed-order fractional differential equation describing forced
oscillations of a body attached to a free end of a light viscoelastic rod.
Explicit form of a solution for a class of linear viscoelastic solids is
given in terms of a convolution integral. Restrictions on storage and loss
moduli following from the Second Law of Thermodynamics play the crucial role
in establishing the form of the solution. Some previous results are shown to
be special cases of the present analysis.

\bigskip

\end{abstract}

\section{Introduction}

We study the initial value problem given by
\begin{gather}
\int_{0}^{1}\phi _{\sigma }\left( \gamma \right) {}_{0}\mathrm{D}%
_{t}^{\gamma }\sigma \left( t\right) \mathrm{d}\gamma =\int_{0}^{1}\phi
_{\varepsilon }\left( \gamma \right) {}_{0}\mathrm{D}_{t}^{\gamma
}\varepsilon \left( t\right) \mathrm{d}\gamma ,\;\;t>0,
\label{sys-lajt-ode-1} \\
-\sigma \left( t\right) +F\left( t\right) =\frac{\mathrm{d}^{2}}{\mathrm{d}%
t^{2}}\varepsilon \left( t\right) ,\;\;t>0,  \label{sys-lajt-ode-2} \\
\sigma \left( 0\right) =0,\;\;\;\varepsilon \left( 0\right) =0,\;\;\;\frac{%
\mathrm{d}}{\mathrm{d}t}\varepsilon \left( 0\right) =0.  \label{IC-lajt-ode}
\end{gather}%
In (\ref{sys-lajt-ode-1}) - (\ref{IC-lajt-ode}) we use $\sigma $ to denote
stress, $\varepsilon $ denotes strain and $F$ denotes the force acting at
the free end of a rod. The left Riemann-Liouville fractional derivative of
order $\gamma \in \left( 0,1\right) $ is defined by
\begin{equation*}
{}_{0}\mathrm{D}_{t}^{\gamma }y\left( t\right) :=\frac{\mathrm{d}}{\mathrm{d}%
t}\left( \frac{t^{-\gamma }}{\Gamma \left( 1-\gamma \right) }\ast y\left(
t\right) \right) ,\;\;t>0,
\end{equation*}%
see \cite{SKM}. The Euler gamma function is denoted by $\Gamma $ and $\ast $
denotes the convolution: $\left( f\ast g\right) \left( t\right)
:=\int_{0}^{t}f\left( \tau \right) g\left( t-\tau \right) \mathrm{d}\tau ,$ $%
t\in
%TCIMACRO{\U{211d} }%
%BeginExpansion
\mathbb{R}
%EndExpansion
,$ if $f,g\in L_{loc}^{1}\left( \mathbf{%
%TCIMACRO{\U{211d} }%
%BeginExpansion
\mathbb{R}
%EndExpansion
}\right) $ and $\limfunc{supp}f,g\subset \left[ 0,\infty \right) $. The
displacement of an arbitrary point of a rod that is at the initial moment at
the position $x$ is
\begin{equation*}
u\left( x,t\right) =x\varepsilon \left( t\right) ,\;\;t>0,\;\;x\in \left[ 0,1%
\right] ,
\end{equation*}%
see \cite{a}.

Equation (\ref{sys-lajt-ode-1}) represents a constitutive equation of a
viscoelastic rod, (\ref{sys-lajt-ode-2}) is equation of motion of a body
attached to a free end of a rod and (\ref{IC-lajt-ode}) represent the
initial conditions. For the study of waves in viscoelastic materials of
fractional type see \cite{Mai-10}.

The derivation of the system (\ref{sys-lajt-ode-1}) - (\ref{IC-lajt-ode}) is
given in \cite{a}, where the special case of (\ref{sys-lajt-ode-1}) is
considered when the constitutive functions $\phi _{\sigma }$ and $\phi
_{\varepsilon }$ have the form
\begin{equation*}
\phi _{\sigma }\left( \gamma \right) =a^{\gamma }\;\;\text{and}\;\;\phi
_{\varepsilon }\left( \gamma \right) =c\,\delta \left( \gamma \right)
+b^{\gamma },\;\;0<a\leq b,\;c>0,
\end{equation*}
where $\delta $ is the Dirac delta distribution. In the present
work we allow constitutive functions (or distributions) $\phi
_{\sigma }$ and $\phi _{\varepsilon }$ to be arbitrary satisfying
Condition \ref{cond-1} and Assumption \ref{cond}. To be physically
admissible (\ref{sys-lajt-ode-1}) must satisfy two conditions: one
being mathematical, the other being physical. The first condition
that is essential on the level of generality treated in the work
requires that to real forcing ($\sigma $ is a
real-valued function of real variable) there corresponds a real response $%
\varepsilon $ ($\varepsilon $ is a real-valued function of real variable).
This mathematical condition, that we believe is new, is precisely formulated
as $\left( i\right) $ of Condition \ref{cond-1}. The second condition is
physical and requires that in any closed deformation cycle there is a
dissipation of energy. This is the formulation of the Second Law of
Thermodynamics for isothermal processes. This condition is stated in its
equivalent form as $\left( ii\right) $ of Condition \ref{cond-1}.

\section{Analysis of the problem}

In the following we use the Laplace transform method. The Laplace transform
is defined by
\begin{equation*}
\tilde{f}\left( s\right) =\mathcal{L}\left[ f\left( t\right) \right] \left(
s\right) :=\int_{0}^{\infty }f\left( t\right) \mathrm{e}^{-st}\mathrm{d}%
t,\;\;\func{Re}s>k,
\end{equation*}%
where $f\in L_{loc}^{1}\left( \mathbf{%
%TCIMACRO{\U{211d} }%
%BeginExpansion
\mathbb{R}
%EndExpansion
}\right) ,$ $f\equiv 0$ in $\left( -\infty ,0\right] $ and $\left\vert
f\left( t\right) \right\vert \leq c\mathrm{e}^{kt},$ $t>0,$ for some $k>0.\ $%
We denote by $\mathcal{S}^{\prime }\left(
%TCIMACRO{\U{211d} }%
%BeginExpansion
\mathbb{R}
%EndExpansion
\right) $ the space of tempered distributions on $%
%TCIMACRO{\U{211d} }%
%BeginExpansion
\mathbb{R}
%EndExpansion
,$ while $\mathcal{S}_{+}^{\prime }$ is its subspace consisting of tempered
distributions with support $\left[ 0,\infty \right) .$ We refer to \cite%
{vlad} for the properties of this space as well as for the Laplace transform
within it. We also use $C\left( \left[ 0,1\right] ,\mathcal{S}_{+}^{\prime
}\right) $ to denote the space of continuous functions on $\left[ 0,1\right]
$ with the values in $\mathcal{S}_{+}^{\prime }.$

Applying formally the Laplace transform to (\ref{sys-lajt-ode-1}) - (\ref%
{IC-lajt-ode}) we get%
\begin{gather}
\tilde{\sigma}\left( s\right) \int_{0}^{1}\phi _{\sigma }\left( \gamma
\right) s^{\gamma }\mathrm{d}\gamma =\tilde{\varepsilon}\left( s\right)
\int_{0}^{1}\phi _{\varepsilon }\left( \gamma \right) s^{\gamma }\mathrm{d}%
\gamma ,\;\;s\in D,  \label{S-LT-2} \\
\tilde{\sigma}\left( s\right) +s^{2}\tilde{\varepsilon}\left( s\right) =%
\tilde{F}\left( s\right) ,\;\;s\in D.  \label{S-LT-3}
\end{gather}%
By (\ref{S-LT-2}) we have
\begin{equation}
\tilde{\sigma}\left( x,s\right) =\frac{1}{M^{2}\left( s\right) }\tilde{%
\varepsilon}\left( x,s\right) ,\;\;s\in D,  \label{sigma-tilda}
\end{equation}%
where
\begin{equation}
M\left( s\right) :=\sqrt{\frac{\int_{0}^{1}\phi _{\sigma }\left( \gamma
\right) s^{\gamma }\mathrm{d}\gamma }{\int_{0}^{1}\phi _{\varepsilon }\left(
\gamma \right) s^{\gamma }\mathrm{d}\gamma }},\;\;s\in D\subset
%TCIMACRO{\U{2102} }%
%BeginExpansion
\mathbb{C}
%EndExpansion
.  \label{M}
\end{equation}%
Note that for $s=\mathrm{i}\omega $ we obtain the complex modulus (see \cite%
{b-t})
\begin{equation}
E\left( \omega \right) =E^{\prime }\left( \omega \right) +\mathrm{i}%
E^{\prime \prime }\left( \omega \right) :=\frac{1}{M^{2}\left( \mathrm{i}%
\omega \right) }=\frac{\int_{0}^{1}\phi _{\varepsilon }\left( \gamma \right)
\left( \mathrm{i}\omega \right) ^{\gamma }\mathrm{d}\gamma }{%
\int_{0}^{1}\phi _{\sigma }\left( \gamma \right) \left( \mathrm{i}\omega
\right) ^{\gamma }\mathrm{d}\gamma },\;\;\omega \in \left( 0,\infty \right) .
\label{ep-es}
\end{equation}%
Functions $E^{\prime }$ and $E^{\prime \prime }$ are real-valued and
represent the storage and loss modulus, respectively.

By (\ref{S-LT-3}) and (\ref{sigma-tilda}) we have%
\begin{equation}
\tilde{\varepsilon}\left( s\right) =\tilde{F}\left( s\right) \tilde{P}\left(
s\right) ,\;\;\;\;\tilde{\sigma}\left( s\right) =\tilde{F}\left( s\right)
\tilde{Q}\left( s\right) ,\;\;s\in D,  \label{eps-sig}
\end{equation}%
where%
\begin{equation}
\tilde{P}\left( s\right) :=\frac{M^{2}\left( s\right) }{1+\left( sM\left(
s\right) \right) ^{2}},\;\;\;\;\tilde{Q}\left( s\right) :=\frac{1}{1+\left(
sM\left( s\right) \right) ^{2}},\;\;s\in D,  \label{P-l}
\end{equation}%
and $M$ is given by (\ref{M}).

Formally, by inverting the Laplace transform in (\ref{eps-sig}), we obtain%
\begin{equation}
\varepsilon \left( t\right) =F\left( t\right) \ast P\left( t\right)
,\;\;\;\;\sigma \left( t\right) =F\left( t\right) \ast Q\left( t\right)
,\;\;t>0.  \label{e-s}
\end{equation}

We discuss the restrictions that $\phi _{\sigma }$ and $\phi _{\varepsilon }$
must satisfy. The first restriction follows from the fact that $P$ and $Q$
must be real-valued functions, so that strain $\varepsilon $ and stress $%
\sigma ,$ given by (\ref{e-s}), are real. The second condition imposes the
Second Law of Thermodynamics, which requires that (in isothermal case) the
dissipation work must be positive. Mathematically, these conditions read as
follows.

\begin{condition}
\label{cond-1}\qquad

\begin{enumerate}
\item[$\left( i\right) $] There exists $x_{0}\in \mathbb{R}$ such that
\begin{equation*}
M\left( x\right) =\sqrt{\frac{\int_{0}^{1}\phi _{\sigma }\left( \gamma
\right) x^{\gamma }\mathrm{d}\gamma }{\int_{0}^{1}\phi _{\varepsilon }\left(
\gamma \right) x^{\gamma }\mathrm{d}\gamma }}\in \mathbb{R},\;\;\text{for
all }x>x_{0}.
\end{equation*}

\item[$\left( ii\right) $] For all $\omega \in \left( 0,\infty \right) $ we
have
\begin{equation*}
E^{\prime }\left( \omega \right) \geq 0,\;\;\;\;E^{\prime \prime }\left(
\omega \right) \geq 0,
\end{equation*}%
where $E^{\prime }$ and $E^{\prime \prime }$ are storage and loss moduli,
respectively, given by (\ref{ep-es}), see \cite{b-t}.
\end{enumerate}
\end{condition}

The motivation for $\left( i\right) $ of Condition \ref{cond-1} follows from
the following theorem of Doetsch.

\begin{theorem}[{\protect\cite[p. 293, Satz 2]{dec}}]
\label{dec}Let $f(s)=\mathcal{L}[F](s),$ $\func{Re}s>x_{0}\in
%TCIMACRO{\U{211d} }%
%BeginExpansion
\mathbb{R}
%EndExpansion
,$ be real-valued on the real half-line $s\in \left( x_{0},\infty \right) .$
Then function $F$ is real-valued almost everywhere.

Alternatively, if $f$ is real-valued at a sequence of equidistant points on
the real axis, then function $F$ is real-valued almost everywhere.
\end{theorem}

If $\phi _{\sigma }$ and $\phi _{\varepsilon }$ are such that $\left(
i\right) $ of Condition \ref{cond-1} is satisfied, Theorem \ref{dec} ensures
that inversions of (\ref{eps-sig}) with (\ref{P-l}), given by (\ref{e-s}),
are real. As it is well-known, \cite{b-t}, $\left( ii\right) $ of Condition %
\ref{cond-1} guarantees that the Second Law of Thermodynamics for the
isothermal process is satisfied. We shall see from Proposition \ref{pr-l}
below that Condition \ref{cond-1} along with an additional assumption on the
asymptotics of $M$ (Assumption \ref{cond} below) guarantees that the poles
of the solution kernel in the Laplace domain (\ref{P-l}) belong to the left
complex half-plane. In this case the amplitude of the solution decreases
with the time; this is a characteristic behavior for a dissipative process.

\begin{remark}
\qquad

\begin{enumerate}
\item[$\left( i\right) $] Condition \ref{cond-1} is satisfied for the
fractional Zener model%
\begin{equation}
\left( 1+a\,{}_{0}\mathrm{D}_{t}^{\alpha }\right) \sigma \left( t\right)
=\left( 1+b\,{}_{0}\mathrm{D}_{t}^{\alpha }\right) \varepsilon \left(
t\right)   \label{Zener}
\end{equation}%
and for the distributed-order model%
\begin{equation}
\int_{0}^{1}a^{\gamma }\,{}_{0}\mathrm{D}_{t}^{\gamma }\sigma \left(
t\right) \mathrm{d}\gamma =\int_{0}^{1}b^{\gamma }\,{}_{0}\mathrm{D}%
_{t}^{\gamma }\varepsilon \left( t\right) \mathrm{d}\gamma ,  \label{A-B}
\end{equation}%
In those case the function $M$ takes the following forms%
\begin{eqnarray}
M\left( s\right)  &=&\sqrt{\frac{1+as^{\alpha }}{1+bs^{\alpha }}},\;\;s\in
%TCIMACRO{\U{2102} }%
%BeginExpansion
\mathbb{C}
%EndExpansion
\backslash \left( -\infty ,0\right] ,\;0<a\leq b,\;\alpha \in \left(
0,1\right) ,  \label{M-Zener} \\
M\left( s\right)  &=&\sqrt{\frac{\ln \left( bs\right) }{\ln \left( as\right)
}\frac{as-1}{bs-1}},\;\;s\in
%TCIMACRO{\U{2102} }%
%BeginExpansion
\mathbb{C}
%EndExpansion
\backslash \left( -\infty ,0\right] ,\;0<a\leq b,  \label{M-a-b}
\end{eqnarray}%
respectively. Putting $s=x\in
%TCIMACRO{\U{211d} }%
%BeginExpansion
\mathbb{R}
%EndExpansion
$ in the previous expressions we see that $\left( i\right) $ of Condition %
\ref{cond-1} is satisfied with $x_{0}=0$, while $0<a\leq b$ ensures that $%
\left( ii\right) $ of Condition \ref{cond-1} is satisfied, see \cite%
{AKOZ,b-t}.

\item[$\left( ii\right) $] In general, stress $\sigma $ as a function of a
real-valued strain $\varepsilon $ may not be real-valued. In \cite{AKP} we
have such a situation, because $\left( i\right) $ of Condition \ref{cond-1}
is not satisfied. This shows the importance of $\left( i\right) $ of
Condition \ref{cond-1}.
\end{enumerate}
\end{remark}

\section{Inversion of the Laplace transforms}

In order to obtain $\varepsilon $ and $\sigma ,$ by (\ref{eps-sig}) we have
to determine functions $P$ and $Q,$ i.e., to invert the Laplace transform in
(\ref{P-l}). We need an additional assumption on the function $M,$ given by (%
\ref{M}).

\begin{assumption}
\label{cond}Let $M$ be of the form%
\begin{equation*}
M\left( s\right) =r\left( s\right) +\mathrm{i}h\left( s\right) ,\;\;\text{as}%
\;\;\left\vert s\right\vert \rightarrow \infty ,
\end{equation*}%
and suppose that%
\begin{equation*}
\lim_{\left\vert s\right\vert \rightarrow \infty }r\left( s\right)
=c_{\infty }>0,\;\;\lim_{\left\vert s\right\vert \rightarrow \infty }h\left(
s\right) =0,\;\;\lim_{\left\vert s\right\vert \rightarrow 0}M\left( s\right)
=c_{0},
\end{equation*}%
for some constants$\;c_{\infty },c_{0}>0.$
\end{assumption}

Assumption \ref{cond} is motivated by the fractional Zener (\ref{Zener}) and
distributed-order model (\ref{A-B}). We note that both of these models
describe the viscoelastic solid-like body. For the both models mentioned
above we have $c_{\infty }=\sqrt{\frac{a}{b}},$ and $c_{0}=1.$

\begin{proposition}
\label{pr-l}Let $M$ satisfy Condition \ref{cond-1} with $x_{0}=0$ and
Assumption \ref{cond}. Let
\begin{equation}
f\left( s\right) :=1+\left( sM\left( s\right) \right) ^{2},\;\;s\in
%TCIMACRO{\U{2102} }%
%BeginExpansion
\mathbb{C}
%EndExpansion
.  \label{polovi-l}
\end{equation}%
Then $f$ has two different zeros: $s_{0}$ and its complex conjugate $\bar{s}%
_{0},$ located in the left complex half-plane ($\func{Re}s<0$). The
multiplicity of each zero is one.
\end{proposition}

\begin{proof}
Since $\bar{s}M\left( \bar{s}\right) =\overline{sM\left( s\right) }$ (bar
denotes the complex conjugation), it is clear from (\ref{polovi-l}) that if $%
s_{0}$ is zero of (\ref{polovi-l}), then its complex conjugate $\bar{s}_{0}$
is also.

We use the argument principle in order to show that there are no zeros of (%
\ref{polovi-l}) in the upper right complex half-plane. Let us consider
contour $\gamma _{R}=\gamma _{R1}\cup \gamma _{R2}\cup \gamma _{R3}\cup
\gamma _{R4},$ presented in Figure \ref{fig-lr}.
\begin{figure}[h]
\centering
\includegraphics[scale=0.65]{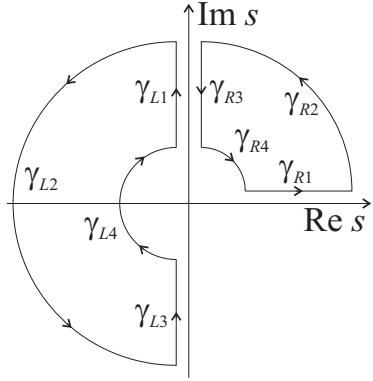}
\caption{Integration contours $\protect\gamma _{L}$ and $\protect\gamma _{R}$%
.}
\label{fig-lr}
\end{figure}
Let $\gamma _{R1}:s=x,$ $x\in \left[ r,R\right] ,$ where $r$ is the radius
of the inner quarter of circle and $R$ is the radius of the outer quarter of
circle. Then, by $\left( i\right) $ of Condition \ref{cond-1} (with $x_{0}=0$%
) we have $\func{Re}f\left( x\right) >0$ and $\func{Im}f\left( x\right)
\equiv 0.$ Assumption \ref{cond} implies $\lim_{r\rightarrow 0}f\left(
x\right) =1$ and $\lim_{R\rightarrow \infty }f\left( x\right) =\infty .$
Hence, $\Delta \arg f\left( s\right) =0$ for $s\in \gamma _{R1},$ as $%
r\rightarrow 0,$ $R\rightarrow \infty .$ Let $\gamma _{R2}:s=R\mathrm{e}^{%
\mathrm{i}\varphi },$ $\varphi \in \left[ 0,\frac{\pi }{2}\right] .$ Then,
by Assumption \ref{cond}, for $\varphi =0$ we have
\begin{equation*}
\func{Re}f\left( R\right) \approx c_{\infty }^{2}R^{2}\rightarrow \infty \;\;%
\text{and}\;\;\func{Im}f\left( R\right) \rightarrow 0,\;\;\text{as}%
\;\;R\rightarrow \infty .
\end{equation*}%
For $\varphi =\frac{\pi }{2}$ Assumption \ref{cond} implies%
\begin{equation}
\func{Re}f\left( R\mathrm{e}^{\mathrm{i}\frac{\pi }{2}}\right) \approx
-c_{\infty }^{2}R^{2}\rightarrow -\infty \;\;\text{and}\;\;\func{Im}f\left( R%
\mathrm{e}^{\mathrm{i}\frac{\pi }{2}}\right) \rightarrow 0,\;\;\text{as}%
\;\;R\rightarrow \infty .  \label{i-pi-pola}
\end{equation}%
By using Assumption \ref{cond} in (\ref{polovi-l}) we obtain%
\begin{equation}
\func{Im}f\left( R\mathrm{e}^{\mathrm{i}\varphi }\right) \approx c_{\infty
}^{2}R^{2}\sin \left( 2\varphi \right) \geq 0,\;\;\varphi \in \left[ 0,\frac{%
\pi }{2}\right] ,\;\;\text{as}\;\;R\rightarrow \infty .  \label{imre}
\end{equation}%
Therefore, $\Delta \arg f\left( s\right) =\pi $ for $s\in \gamma _{R2},$ as $%
R\rightarrow \infty .$ Let $\gamma _{R3}:s=\omega \mathrm{e}^{\mathrm{i}%
\frac{\pi }{2}}=\mathrm{i}\omega ,$ $\omega \in \left[ r,R\right] .$ Then we
have%
\begin{equation}
f\left( \mathrm{i}\omega \right) =1-\omega ^{2}M^{2}\left( \mathrm{i}\omega
\right) ,\;\;\omega \in \left[ r,R\right] .  \label{f-od-ix}
\end{equation}%
Using (\ref{ep-es}) in (\ref{f-od-ix}) we obtain%
\begin{equation*}
f\left( \mathrm{i}\omega \right) =1-\omega ^{2}\frac{E^{\prime }\left(
\omega \right) -\mathrm{i}E^{\prime \prime }\left( \omega \right) }{\left(
E^{\prime }\left( \omega \right) \right) ^{2}+\left( E^{\prime \prime
}\left( \omega \right) \right) ^{2}},\;\;\omega \in \left[ r,R\right] .
\end{equation*}%
Thus,%
\begin{equation}
\func{Im}f\left( \mathrm{i}\omega \right) =\omega ^{2}\frac{E^{\prime \prime
}\left( \omega \right) }{\left( E^{\prime }\left( \omega \right) \right)
^{2}+\left( E^{\prime \prime }\left( \omega \right) \right) ^{2}}%
>0,\;\;\omega \in \left[ r,R\right] ,  \label{im-po-im}
\end{equation}%
due to $\left( ii\right) $ of Condition \ref{cond-1}. Moreover, Assumption %
\ref{cond} applied to (\ref{f-od-ix}) implies%
\begin{eqnarray}
\func{Re}f\left( \mathrm{i}\omega \right)  &\approx &-c_{\infty }^{2}\omega
^{2}\rightarrow -\infty \;\;\text{and}\;\;\func{Im}f\left( \mathrm{i}\omega
\right) \rightarrow 0,\;\;\text{as}\;\;R\rightarrow \infty ,
\label{re-po-im-u-besk} \\
\func{Re}f\left( \mathrm{i}\omega \right)  &\approx &1-c_{0}^{2}\omega
^{2}\rightarrow 1\;\;\text{and}\;\;\func{Im}f\left( \mathrm{i}\omega \right)
\rightarrow 0,\;\;\text{as}\;\;r\rightarrow 0.  \label{re-po-im-u-nuli}
\end{eqnarray}%
We conclude that $\Delta \arg f\left( s\right) =-\pi $ for $s\in \gamma
_{R3},$ as $r\rightarrow 0,$ $R\rightarrow \infty .$ Let $\gamma _{R4}:s=r%
\mathrm{e}^{\mathrm{i}\varphi },$ $\varphi \in \left[ 0,\frac{\pi }{2}\right]
.$ Assumption \ref{cond}, for $\varphi =0$ and $\varphi =\frac{\pi }{2}$
implies
\begin{eqnarray}
&&\func{Re}f\left( r\right) \approx 1+c_{0}^{2}r^{2}\rightarrow 1\;\;\text{%
and}\;\;\func{Im}f\left( r\right) \rightarrow 0,\;\;\text{as}%
\;\;r\rightarrow 0,  \notag \\
&&\func{Re}f\left( r\mathrm{e}^{\mathrm{i}\frac{\pi }{2}}\right) \approx
1-c_{0}^{2}r^{2}\rightarrow 1\;\;\text{and}\;\;\func{Im}f\left( r\mathrm{e}^{%
\mathrm{i}\frac{\pi }{2}}\right) \rightarrow 0,\;\;\text{as}\;\;r\rightarrow
0,  \label{i-pi-pola-1}
\end{eqnarray}%
as well as%
\begin{equation}
\func{Re}f\left( r\mathrm{e}^{\mathrm{i}\varphi }\right) \approx
1+c_{0}^{2}r^{2}\cos \left( 2\varphi \right) \rightarrow 1\;\;\text{and}\;\;%
\func{Im}f\left( r\mathrm{e}^{\mathrm{i}\varphi }\right) \approx
c_{0}^{2}r^{2}\sin \left( 2\varphi \right) \rightarrow 0,  \label{re-im}
\end{equation}%
for $\varphi \in \left[ 0,\frac{\pi }{2}\right] ,$ as $r\rightarrow 0.$ We
see that $\Delta \arg f\left( s\right) =0$ for $s\in \gamma _{R4},$ as $%
r\rightarrow 0.$ Thus, we conclude that
\begin{equation*}
\Delta \arg f\left( s\right) =0\;\;\text{for}\;\;s\in \gamma _{R},\;\;\text{%
as}\;\;r\rightarrow 0,\;\;R\rightarrow \infty .
\end{equation*}%
By the argument principle and that fact that the zeros are complex
conjugated we conclude that $f$ has no zeros in the right complex half-plane.

We shall again use the argument principle and show that there are two zeros
of (\ref{polovi-l}) in the left complex half-plane. Let us consider contour $%
\gamma _{L}=\gamma _{L1}\cup \gamma _{L2}\cup \gamma _{L3}\cup \gamma _{L4},$
presented in Figure \ref{fig-lr}. Let $\gamma _{L1}:s=\omega \mathrm{e}^{%
\mathrm{i}\frac{\pi }{2}}=\mathrm{i}\omega ,$ $\omega \in \left[ r,R\right] .
$ Then (\ref{im-po-im}), (\ref{re-po-im-u-besk}) and (\ref{re-po-im-u-nuli})
hold. Thus, $\Delta \arg f\left( s\right) =\pi $ for $s\in \gamma _{L1},$ as
$r\rightarrow 0,$ $R\rightarrow \infty .$ Let $\gamma _{L2}:s=R\mathrm{e}^{%
\mathrm{i}\varphi },$ $\varphi \in \left[ \frac{\pi }{2},\frac{3\pi }{2}%
\right] .$ Then, for $\varphi =\frac{\pi }{2}$ we have that (\ref{i-pi-pola}%
) holds. For $\varphi =\pi $ Assumption \ref{cond} implies
\begin{equation*}
\func{Re}f\left( R\mathrm{e}^{\mathrm{i}\pi }\right) \approx c_{\infty
}^{2}R^{2}\rightarrow \infty \;\;\text{and}\;\;\func{Im}f\left( R\mathrm{e}^{%
\mathrm{i}\pi }\right) \rightarrow 0,\;\;\text{as}\;\;R\rightarrow \infty ,
\end{equation*}%
while for $\varphi =\frac{3\pi }{2}$ it implies
\begin{equation*}
\func{Re}f\left( R\mathrm{e}^{\mathrm{i}\frac{3\pi }{2}}\right) \approx
-c_{\infty }^{2}R^{2}\rightarrow -\infty \;\;\text{and}\;\;\func{Im}f\left( R%
\mathrm{e}^{\mathrm{i}\frac{3\pi }{2}}\right) \rightarrow 0,\;\;\text{as}%
\;\;R\rightarrow \infty .
\end{equation*}%
By (\ref{imre}), we have
\begin{eqnarray*}
\func{Im}f\left( R\mathrm{e}^{\mathrm{i}\varphi }\right)  &\approx
&c_{\infty }^{2}R^{2}\sin \left( 2\varphi \right) \leq 0,\;\;\varphi \in %
\left[ \frac{\pi }{2},\pi \right) ,\;\;\text{as}\;\;R\rightarrow \infty , \\
\func{Im}f\left( R\mathrm{e}^{\mathrm{i}\varphi }\right)  &\approx
&c_{\infty }^{2}R^{2}\sin \left( 2\varphi \right) \geq 0,\;\;\varphi \in
\left( \pi ,\frac{3\pi }{2}\right] ,\;\;\text{as}\;\;R\rightarrow \infty .
\end{eqnarray*}%
We conclude that $\Delta \arg f\left( s\right) =2\pi $ for $s\in \gamma
_{L2},$ as $R\rightarrow \infty .$ Let $\gamma _{L3}:s=\omega \mathrm{e}^{%
\mathrm{i}\frac{3\pi }{2}}=-\mathrm{i}\omega ,$ $\omega \in \left[ r,R\right]
.$ By (\ref{polovi-l}) we have%
\begin{equation}
f\left( -\mathrm{i}\omega \right) =1-\omega ^{2}\overline{M^{2}\left(
\mathrm{i}\omega \right) },\;\;\omega \in \left[ r,R\right] ,
\label{f-od-minus-ix}
\end{equation}%
since $M\left( \bar{s}\right) =\overline{M\left( s\right) }.$ Using (\ref%
{ep-es}) in (\ref{f-od-minus-ix}) we obtain%
\begin{equation*}
f\left( -\mathrm{i}\omega \right) =1-\omega ^{2}\frac{E^{\prime }\left(
\omega \right) +\mathrm{i}E^{\prime \prime }\left( \omega \right) }{\left(
E^{\prime }\left( \omega \right) \right) ^{2}+\left( E^{\prime \prime
}\left( \omega \right) \right) ^{2}},\;\;\omega \in \left[ r,R\right] .
\end{equation*}%
Thus,%
\begin{equation*}
\func{Im}f\left( -\mathrm{i}\omega \right) =-\omega ^{2}\frac{E^{\prime
\prime }\left( \omega \right) }{\left( E^{\prime }\left( \omega \right)
\right) ^{2}+\left( E^{\prime \prime }\left( \omega \right) \right) ^{2}}%
<0,\;\;\omega \in \left[ r,R\right] ,
\end{equation*}%
due to $\left( ii\right) $ of Condition \ref{cond-1}. Note that Assumption %
\ref{cond} applied to (\ref{f-od-minus-ix}) implies%
\begin{eqnarray*}
\func{Re}f\left( -\mathrm{i}\omega \right)  &\approx &-c_{\infty }^{2}\omega
^{2}\rightarrow -\infty \;\;\text{and}\;\;\func{Im}f\left( \mathrm{i}\omega
\right) \rightarrow 0,\;\;\text{as}\;\;R\rightarrow \infty , \\
\func{Re}f\left( -\mathrm{i}\omega \right)  &\approx &1-c_{0}^{2}\omega
^{2}\rightarrow 1\;\;\text{and}\;\;\func{Im}f\left( \mathrm{i}\omega \right)
\rightarrow 0,\;\;\text{as}\;\;r\rightarrow 0.
\end{eqnarray*}%
We conclude that $\Delta \arg f\left( s\right) =\pi $ for $s\in \gamma _{L3},
$ as $r\rightarrow 0,$ $R\rightarrow \infty .$ Let $\gamma _{L4}:s=r\mathrm{e%
}^{\mathrm{i}\varphi },$ $\varphi \in \left[ \frac{\pi }{2},\frac{3\pi }{2}%
\right] .$ Then, for $\varphi =\frac{\pi }{2}$ we have that (\ref%
{i-pi-pola-1}) holds. For $\varphi =\frac{3\pi }{2}$ Assumption \ref{cond}
implies%
\begin{equation*}
\func{Re}f\left( r\mathrm{e}^{\mathrm{i}\frac{3\pi }{2}}\right) \approx
1-c_{0}^{2}r^{2}\rightarrow 1\;\;\text{and}\;\;\func{Im}f\left( r\mathrm{e}^{%
\mathrm{i}\frac{3\pi }{2}}\right) \rightarrow 0,\;\;\text{as}%
\;\;r\rightarrow 0.
\end{equation*}%
We have that also (\ref{re-im}) holds for $\varphi \in \left[ \frac{\pi }{2},%
\frac{3\pi }{2}\right] ,$ as $r\rightarrow 0,$ so that $\Delta \arg f\left(
s\right) =0$ for $s\in \gamma _{L4},$ as $r\rightarrow 0.$ Thus, the
conclusion is that%
\begin{equation*}
\Delta \arg f\left( s\right) =4\pi \;\;\text{for}\;\;s\in \gamma _{L},\;\;%
\text{as}\;\;r\rightarrow 0,\;\;R\rightarrow \infty .
\end{equation*}%
This implies that $f$ has two zeros in the left complex half-plane.
\end{proof}

The following theorem is related to the existence of solutions to system (%
\ref{sys-lajt-ode-1}) - (\ref{IC-lajt-ode}).

\begin{theorem}
\label{thm}Let $M$ satisfy Condition \ref{cond-1} with $x_{0}=0$ and
Assumption \ref{cond}. Let $F\in \mathcal{S}_{+}^{\prime }.$

\begin{enumerate}
\item[$\left( i\right) $] The displacement $u\ $as a part of the solution to
(\ref{sys-lajt-ode-1}) - (\ref{IC-lajt-ode}), is given by%
\begin{equation}
u\left( x,t\right) =x\varepsilon \left( t\right) ,\;\;\text{where}%
\;\;\varepsilon \left( t\right) =F\left( t\right) \ast P\left( t\right)
,\;\;x\in \left[ 0,1\right] ,\;t>0,  \label{u-eps}
\end{equation}%
and%
\begin{eqnarray}
P\left( t\right)  &=&\frac{1}{\pi }\dint\nolimits_{0}^{\infty }\func{Im}%
\left( \frac{M^{2}\left( q\mathrm{e}^{-\mathrm{i}\pi }\right) }{1+\left(
qM\left( q\mathrm{e}^{-\mathrm{i}\pi }\right) \right) ^{2}}\right) \mathrm{e}%
^{-qt}\mathrm{d}q+2\func{Re}\left( \func{Res}\left( \tilde{P}\left( s\right)
\mathrm{e}^{st},s_{0}\right) \right) ,\;\;t>0,  \notag \\
&&  \label{P1} \\
P\left( t\right)  &=&0,\;\;t<0.  \notag
\end{eqnarray}%
The residue term in (\ref{P1}) is given by%
\begin{equation}
\func{Res}\left( \tilde{P}\left( s\right) \mathrm{e}^{st},s_{0}\right) =%
\left[ \frac{M^{2}\left( s\right) }{\frac{\mathrm{d}}{\mathrm{d}s}f\left(
s\right) }\mathrm{e}^{st}\right] _{s=s_{0}},\;\;t>0,  \label{res-P}
\end{equation}%
where $f$ is given by (\ref{polovi-l}) and $s_{0}$ is the zero of $f.$
Function $P$ is real-valued, continuous on $\left[ 0,\infty \right) $ and $%
u\in C\left( \left[ 0,1\right] ,\mathcal{S}_{+}^{\prime }\right) .$
Moreover, if $F$ is locally integrable on $%
%TCIMACRO{\U{211d} }%
%BeginExpansion
\mathbb{R}
%EndExpansion
$ (and equals zero on $\left( -\infty ,0\right] $) then $u\in C\left( \left[
0,1\right] \times \left[ 0,\infty \right) \right) .$

\item[$\left( ii\right) $] The stress $\sigma $ as a part of the solution to
(\ref{sys-lajt-ode-1}) - (\ref{IC-lajt-ode}) is given by%
\begin{equation*}
\sigma \left( t\right) =F\left( t\right) \ast Q\left( t\right) ,\;\;x\in
\left[ 0,1\right] ,\;t>0,
\end{equation*}%
where%
\begin{eqnarray}
Q\left( t\right)  &=&\frac{1}{\pi }\dint\nolimits_{0}^{\infty }\func{Im}%
\left( \frac{1}{1+\left( qM\left( q\mathrm{e}^{-\mathrm{i}\pi }\right)
\right) ^{2}}\right) \mathrm{e}^{-qt}\mathrm{d}q+2\func{Re}\left( \func{Res}%
\left( \tilde{Q}\left( s\right) \mathrm{e}^{st},s_{0}\right) \right)
,\;\;t>0,  \notag \\
&&  \label{Q11} \\
Q\left( t\right)  &=&0,\;\;t<0.  \notag
\end{eqnarray}%
The residue term in (\ref{Q11}) is given by%
\begin{equation}
\func{Res}\left( \tilde{Q}\left( s\right) \mathrm{e}^{st},s_{0}\right) =%
\left[ \frac{1}{\frac{\mathrm{d}}{\mathrm{d}s}f\left( s\right) }\mathrm{e}%
^{st}\right] _{s=s_{0}},\;\;t>0,  \label{res-Q}
\end{equation}%
where $f$ is given by (\ref{polovi-l}) and $s_{0}$ is the zero of $f.$
Function $Q$ is real-valued, continuous on $\left[ 0,\infty \right) .$
Moreover, if $F$ is locally integrable on $%
%TCIMACRO{\U{211d} }%
%BeginExpansion
\mathbb{R}
%EndExpansion
$ (and equals zero on $\left( -\infty ,0\right] $) then $\sigma \in C\left( %
\left[ 0,\infty \right) \right) .$
\end{enumerate}
\end{theorem}

\begin{proof}
We prove $\left( i\right) $ of Theorem. Function $P$ is real-valued by $%
\left( i\right) $ of Condition \ref{cond-1}. We calculate $P\left( t\right) ,
$ $t\in
%TCIMACRO{\U{211d} }%
%BeginExpansion
\mathbb{R}
%EndExpansion
,$ by the integration over the contour given in Figure \ref{fig}.
\begin{figure}[h]
\centering
\includegraphics[scale=0.65]{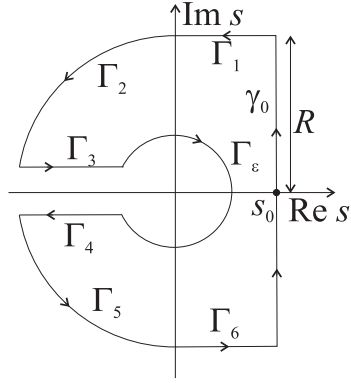}
\caption{Integration contour $\Gamma $}
\label{fig}
\end{figure}
The Cauchy residues theorem yields
\begin{equation}
\oint\nolimits_{\Gamma }\tilde{P}\left( s\right) \mathrm{e}^{st}\mathrm{d}%
s=2\pi \mathrm{i}\left( \func{Res}\left( \tilde{P}\left( s\right) \mathrm{e}%
^{st},s_{0}\right) +\func{Res}\left( \tilde{P}\left( s\right) \mathrm{e}%
^{st},\bar{s}_{0}\right) \right) ,  \label{KF-P}
\end{equation}%
where $\Gamma =\Gamma _{1}\cup \Gamma _{2}\cup \Gamma _{3}\cup \Gamma
_{4}\cup \Gamma _{5}\cup \Gamma _{6}\cup \Gamma _{7}\cup \gamma _{0},$ so
that poles of $\tilde{P},$ given by (\ref{P-l}), lie inside the contour $%
\Gamma .$ Proposition \ref{pr-l} implies that the pole $s_{0},$ and its
complex conjugate $\bar{s}_{0},$ of $\tilde{P}$ are simple. Then the
residues in (\ref{KF-P}) can be calculated using (\ref{res-P}).

Now, we calculate the integral over $\Gamma $ in (\ref{KF-P}). First, we
consider the integral along contour $\Gamma _{1}=\left\{ s=p+\mathrm{i}R\mid
p\in \left[ 0,s_{0}\right] ,\;R>0\right\} ,$ where $R$ is such that the
poles $s_{0}$ and $\bar{s}_{0}$ lie inside the contour $\Gamma .$ By (\ref%
{P-l}) and Assumption \ref{cond} we have
\begin{equation}
\left\vert \tilde{P}\left( s\right) \right\vert \leq \frac{C}{\left\vert
s\right\vert ^{2}},\;\;\left\vert s\right\vert \rightarrow \infty .
\label{P-s-besk}
\end{equation}%
Using (\ref{P-s-besk}), we calculate the integral over $\Gamma _{1}$ as%
\begin{eqnarray*}
\lim_{R\rightarrow \infty }\left\vert \int\nolimits_{\Gamma _{1}}\tilde{P}%
\left( s\right) \mathrm{e}^{st}\mathrm{d}s\right\vert &\leq
&\lim_{R\rightarrow \infty }\int_{0}^{s_{0}}\left\vert \tilde{P}\left( p+%
\mathrm{i}R\right) \right\vert \left\vert \mathrm{e}^{\left( p+\mathrm{i}%
R\right) t}\right\vert \mathrm{d}p \\
&\leq &C\lim_{R\rightarrow \infty }\int_{0}^{s_{0}}\frac{1}{R^{2}}\mathrm{e}%
^{pt}\mathrm{d}p=0,\;\;t>0.
\end{eqnarray*}%
Similar arguments are valid for the integral along contour $\Gamma _{7}$:%
\begin{equation*}
\lim\limits_{R\rightarrow \infty }\left\vert \int\nolimits_{\Gamma _{7}}%
\tilde{P}\left( s\right) \mathrm{e}^{st}\mathrm{d}s\right\vert =0,\;\;t>0.
\end{equation*}

Next, we consider the integral along contour $\Gamma _{2}.$ By (\ref%
{P-s-besk}) we obtain%
\begin{eqnarray*}
\lim_{R\rightarrow \infty }\left\vert \int\nolimits_{\Gamma _{2}}\tilde{P}%
\left( s\right) \mathrm{e}^{st}\mathrm{d}s\right\vert &\leq
&\lim_{R\rightarrow \infty }\int\nolimits_{\frac{\pi }{2}}^{\pi }\left\vert
\tilde{P}\left( R\mathrm{e}^{\mathrm{i}\phi }\right) \right\vert \left\vert
\mathrm{e}^{Rt\mathrm{e}^{\mathrm{i}\phi }}\right\vert \left\vert \mathrm{i}R%
\mathrm{e}^{\mathrm{i}\phi }\right\vert \mathrm{d}\phi \\
&\leq &C\lim_{R\rightarrow \infty }\int\nolimits_{\frac{\pi }{2}}^{\pi }%
\frac{1}{R}\mathrm{e}^{Rt\cos \phi }\mathrm{d}\phi =0,\;\;t>0,
\end{eqnarray*}%
since $\cos \phi \leq 0$ for $\phi \in \left[ \frac{\pi }{2},\pi \right] .$
Similarly, we have%
\begin{equation*}
\lim\limits_{R\rightarrow \infty }\left\vert \int\nolimits_{\Gamma _{6}}%
\tilde{P}\left( s\right) \mathrm{e}^{st}\mathrm{d}s\right\vert =0,\;\;t>0.
\end{equation*}

Consider the integral along $\Gamma _{4}.$ Let $\left\vert s\right\vert
\rightarrow 0.$ Then, by Assumption \ref{cond}, $M\left( s\right)
\rightarrow c_{0}$ and $sM\left( s\right) \rightarrow 0.$ Hence, from (\ref%
{P-l}) we have
\begin{equation}
\left\vert \tilde{P}\left( s\right) \right\vert \approx \left\vert M\left(
s\right) \right\vert ^{2}\approx c_{0}^{2},\;\;\text{as}\;\;|s|\rightarrow 0.
\label{P za s nula}
\end{equation}%
The integration along contour $\Gamma _{4}$ gives%
\begin{eqnarray*}
\lim_{r\rightarrow 0}\left\vert \int\nolimits_{\Gamma _{4}}\tilde{P}\left(
s\right) \mathrm{e}^{st}\mathrm{d}s\right\vert &\leq &\lim_{r\rightarrow
0}\int\nolimits_{-\pi }^{\pi }\left\vert \tilde{P}\left( r\mathrm{e}^{%
\mathrm{i}\phi }\right) \right\vert \left\vert \mathrm{e}^{rt\mathrm{e}^{%
\mathrm{i}\phi }}\right\vert \left\vert \mathrm{i}r\mathrm{e}^{\mathrm{i}%
\phi }\right\vert \mathrm{d}\phi \\
&\leq &c_{0}^{2}\lim_{r\rightarrow 0}\int\nolimits_{-\pi }^{\pi }r\mathrm{e}%
^{rt\cos \phi }\mathrm{d}\phi =0,\;\;t>0,
\end{eqnarray*}%
where we used (\ref{P za s nula}).

Integrals along $\Gamma _{3},$ $\Gamma _{5}$ and $\gamma _{0}$ give ($t>0$)%
\begin{eqnarray}
\lim_{\substack{ R\rightarrow \infty  \\ r\rightarrow 0}}\int\nolimits_{%
\Gamma _{3}}\tilde{P}\left( s\right) \mathrm{e}^{st}\mathrm{d}s
&=&\int\nolimits_{0}^{\infty }\frac{M^{2}\left( q\mathrm{e}^{\mathrm{i}\pi
}\right) }{1+\left( qM\left( q\mathrm{e}^{\mathrm{i}\pi }\right) \right) ^{2}%
}\mathrm{e}^{-qt}\mathrm{d}q,  \label{P-plus} \\
\lim_{\substack{ R\rightarrow \infty  \\ r\rightarrow 0}}\int\nolimits_{%
\Gamma _{5}}\tilde{P}\left( s\right) \mathrm{e}^{st}\mathrm{d}s
&=&-\int\nolimits_{0}^{\infty }\frac{M^{2}\left( q\mathrm{e}^{-\mathrm{i}\pi
}\right) }{1+\left( qM\left( q\mathrm{e}^{-\mathrm{i}\pi }\right) \right)
^{2}}\mathrm{e}^{-qt}\mathrm{d}q,  \label{P-minus} \\
\lim_{R\rightarrow \infty }\int\nolimits_{\gamma _{0}}\tilde{P}\left(
s\right) \mathrm{e}^{st}\mathrm{d}s &=&2\pi \mathrm{i}P\left( t\right) .
\label{P-pi}
\end{eqnarray}%
We note that (\ref{P-pi}) is valid if the inversion of the Laplace transform
exists, which is true since all the singularities of $\tilde{P}$ are left
from the line $\gamma _{0}$ and the estimates on $\tilde{P}$ over $\gamma
_{0}$ imply the convergence of the integral. Summing up (\ref{P-plus}), (\ref%
{P-minus}) and (\ref{P-pi}) we obtain the left hand side of (\ref{KF-P}) and
finally $P$ in the form given by (\ref{P1}). Analyzing separately%
\begin{equation*}
\frac{1}{\pi }\dint\nolimits_{0}^{\infty }\func{Im}\left( \frac{M^{2}\left( q%
\mathrm{e}^{-\mathrm{i}\pi }\right) }{1+\left( qM\left( q\mathrm{e}^{-%
\mathrm{i}\pi }\right) \right) ^{2}}\right) \mathrm{e}^{-qt}\mathrm{d}%
q,\;\;\;\;2\sum_{n=1}^{\infty }\func{Re}\left( \func{Res}\left( \tilde{P}%
\left( s\right) \mathrm{e}^{st},s_{0}\right) \right) ,
\end{equation*}%
we conclude that both terms appearing in (\ref{P1}) are continuous functions
on $t\in \left[ 0,\infty \right) .$ This implies that $u$ is a continuous
function on $\left[ 0,1\right] \times \left[ 0,\infty \right) .$ From the
uniqueness of the Laplace transform it follows that $u$ is unique. Since $F$
belongs to $\mathcal{S}_{+}^{\prime }$, it follows that%
\begin{equation*}
u\left( x,\cdot \right) =x\left( F\left( \cdot \right) \ast P\left( \cdot
\right) \right) \in \mathcal{S}_{+}^{\prime },
\end{equation*}%
for every $x\in \left[ 0,1\right] $ and $u\in C\left( \left[ 0,1\right] ,%
\mathcal{S}_{+}^{\prime }\right) .$ Moreover, if $F\in L_{loc}^{1}\left( %
\left[ 0,\infty \right) \right) ,$ then $u\in C\left( \left[ 0,1\right]
\times \left[ 0,\infty \right) \right) ,$ since $P$ is continuous.

Now we prove $\left( ii\right) $ of Theorem. Again, $\left( i\right) $ of
Condition \ref{cond-1} ensures that $Q$ is a real-valued function. We
calculate $Q\left( t\right) ,$ $t\in
%TCIMACRO{\U{211d} }%
%BeginExpansion
\mathbb{R}
%EndExpansion
,$ by the integration over the same contour from Figure \ref{fig}. The
Cauchy residues theorem yields
\begin{equation}
\oint\nolimits_{\Gamma }\tilde{Q}\left( s\right) \mathrm{e}^{st}\mathrm{d}%
s=2\pi \mathrm{i}\left( \func{Res}\left( \tilde{Q}\left( s\right) \mathrm{e}%
^{st},s_{0}\right) +\func{Res}\left( \tilde{Q}\left( s\right) \mathrm{e}%
^{st},\bar{s}_{0}\right) \right) ,  \label{KF-Q}
\end{equation}%
so that poles of $\tilde{Q}$ lie inside the contour $\Gamma $. The poles $%
s_{0}$ and $\bar{s}_{0}$ of $\tilde{Q},$ given by (\ref{P-l}) are the same
as for the function $\tilde{P}.$ Since the poles $s_{0}$ and $\bar{s}_{0}$
are simple, the residues in (\ref{KF-Q}) can be calculated using (\ref{res-Q}%
).

Let us calculate the integral over $\Gamma $ in (\ref{KF-Q}). Consider the
integral along contour%
\begin{equation*}
\Gamma _{1}=\left\{ s=p+\mathrm{i}R\mid p\in \left[ 0,s_{0}\right]
,\;R>0\right\} .
\end{equation*}%
By (\ref{P-l}) and Assumption \ref{cond}, we have
\begin{equation}
\tilde{Q}\left( s\right) \leq \frac{C}{\left\vert s\right\vert ^{2}}%
,\;\;\left\vert s\right\vert \rightarrow \infty .  \label{estim}
\end{equation}%
Using (\ref{estim}) we calculate the integral over $\Gamma _{1}$ as
\begin{eqnarray*}
\lim_{R\rightarrow \infty }\left\vert \int\nolimits_{\Gamma _{1}}\tilde{Q}%
\left( s\right) \mathrm{e}^{st}\mathrm{d}s\right\vert &\leq
&\lim_{R\rightarrow \infty }\int_{0}^{s_{0}}\left\vert \tilde{Q}\left( p+%
\mathrm{i}R\right) \right\vert \left\vert \mathrm{e}^{\left( p+\mathrm{i}%
R\right) t}\right\vert \mathrm{d}p \\
&\leq &C\lim_{R\rightarrow \infty }\int_{0}^{s_{0}}\frac{1}{R^{2}}\mathrm{e}%
^{pt}\mathrm{d}p=0,\;\;t>0,
\end{eqnarray*}%
Similar arguments are valid for the integral along $\Gamma _{7}.$ Thus, we
have%
\begin{equation*}
\lim\limits_{R\rightarrow \infty }\left\vert \int\nolimits_{\Gamma _{7}}%
\tilde{Q}\left( s\right) \mathrm{e}^{st}\mathrm{d}s\right\vert =0,\;\;t>0.
\end{equation*}

With (\ref{estim}) we have that the integral over $\Gamma _{2}$ becomes
\begin{eqnarray*}
\lim_{R\rightarrow \infty }\left\vert \int\nolimits_{\Gamma _{2}}\tilde{Q}%
\left( s\right) \mathrm{e}^{st}\mathrm{d}s\right\vert &\leq
&\lim_{R\rightarrow \infty }\int\nolimits_{\frac{\pi }{2}}^{\pi }\left\vert
\tilde{Q}\left( R\mathrm{e}^{\mathrm{i}\phi }\right) \right\vert \left\vert
\mathrm{e}^{Rt\mathrm{e}^{\mathrm{i}\phi }}\right\vert \left\vert \mathrm{i}R%
\mathrm{e}^{\mathrm{i}\phi }\right\vert \mathrm{d}\phi \\
&\leq &C\lim_{R\rightarrow \infty }\int\nolimits_{\frac{\pi }{2}}^{\pi }%
\frac{1}{R}\mathrm{e}^{Rt\cos \phi }\mathrm{d}\phi =0,\;\;t>0,
\end{eqnarray*}%
since $\cos \phi \leq 0$ for $\phi \in \left[ \frac{\pi }{2},\pi \right] .$
Similar arguments are valid for the integral along $\Gamma _{6}$:%
\begin{equation*}
\lim\limits_{R\rightarrow \infty }\left\vert \int\nolimits_{\Gamma _{6}}%
\tilde{Q}\left( s\right) \mathrm{e}^{st}\mathrm{d}s\right\vert =0,\;\;t>0.
\end{equation*}

Since $M\left( s\right) \rightarrow c_{0}$ and $sM\left( s\right)
\rightarrow 0$ as $\left\vert s\right\vert \rightarrow 0,$ (\ref{P-l})
implies $\tilde{Q}\left( s\right) \approx 1,$ as $|s|\rightarrow 0,$ so that
the integration along contour $\Gamma _{4}$ gives%
\begin{eqnarray*}
\lim_{r\rightarrow 0}\left\vert \int\nolimits_{\Gamma _{4}}\tilde{Q}\left(
s\right) \mathrm{e}^{st}\mathrm{d}s\right\vert &\leq &\lim_{r\rightarrow
0}\int\nolimits_{-\pi }^{\pi }\left\vert \tilde{Q}\left( r\mathrm{e}^{%
\mathrm{i}\phi }\right) \right\vert \left\vert \mathrm{e}^{rt\mathrm{e}^{%
\mathrm{i}\phi }}\right\vert \left\vert \mathrm{i}r\mathrm{e}^{\mathrm{i}%
\phi }\right\vert \mathrm{d}\phi \\
&\leq &\lim_{r\rightarrow 0}\int\nolimits_{-\pi }^{\pi }r\mathrm{e}^{rt\cos
\phi }\mathrm{d}\phi =0,\;\;t>0.
\end{eqnarray*}

Integrals along $\Gamma _{3},$ $\Gamma _{5}$ and $\gamma _{0}$ give ($t>0$)%
\begin{eqnarray}
\lim_{\substack{ R\rightarrow \infty  \\ r\rightarrow 0}}\int\nolimits_{%
\Gamma _{3}}\tilde{Q}\left( s\right) \mathrm{e}^{st}\mathrm{d}s
&=&\int\nolimits_{0}^{\infty }\frac{1}{1+\left( qM\left( q\mathrm{e}^{%
\mathrm{i}\pi }\right) \right) ^{2}}\mathrm{e}^{-qt}\mathrm{d}q,
\label{Q-plus} \\
\lim_{\substack{ R\rightarrow \infty  \\ r\rightarrow 0}}\int\nolimits_{%
\Gamma _{5}}\tilde{Q}\left( s\right) \mathrm{e}^{st}\mathrm{d}s
&=&-\int\nolimits_{0}^{\infty }\frac{1}{1+\left( qM\left( q\mathrm{e}^{-%
\mathrm{i}\pi }\right) \right) ^{2}}\mathrm{e}^{-qt}\mathrm{d}q,
\label{Q-minus} \\
\lim_{R\rightarrow \infty }\int\nolimits_{\gamma _{0}}\tilde{Q}\left(
s\right) \mathrm{e}^{st}\mathrm{d}s &=&2\pi \mathrm{i}Q\left( t\right) .
\label{Q-pi}
\end{eqnarray}%
By the same arguments as in the proof of $\left( i\right) $ we have that (%
\ref{Q-pi}) is valid if the inversion of the Laplace transform exists. This
is true since all the singularities of $\tilde{Q}$ are left from the line $%
\gamma _{0}$ and appropriate estimates on $\tilde{Q}$ are satisfied. Adding (%
\ref{Q-plus}), (\ref{Q-minus}) and (\ref{Q-pi}) we obtain the left hand side
of (\ref{KF-Q}) and finally $Q$ in the form given by (\ref{Q11}).

Since%
\begin{equation*}
\frac{1}{\pi }\dint\nolimits_{0}^{\infty }\func{Im}\left( \frac{1}{1+\left(
qM\left( q\mathrm{e}^{-\mathrm{i}\pi }\right) \right) ^{2}}\right) \mathrm{e}%
^{-qt}\mathrm{d}q,\;\;\;\;2\func{Res}\left( \tilde{Q}\left( s\right) \mathrm{%
e}^{st},s_{0}\right) ,\;\;t>0,
\end{equation*}%
are continuous, it follows that $Q$ is continuous on $\left[ 0,\infty
\right) .$
\end{proof}

\section{Example}

Suppose that $F$ is harmonic, i.e., $F\left( t\right) =F_{0}\cos \left(
\omega t\right) ,$ and that%
\begin{equation}
\phi _{\sigma }\left( \gamma \right) =\delta \left( \gamma \right)
+a\,\delta \left( \alpha -\gamma \right) ,\;\;\;\;\phi _{\varepsilon }\left(
\gamma \right) =\delta \left( \gamma \right) +b\,\delta \left( \alpha
-\gamma \right)   \label{zener}
\end{equation}%
in (\ref{sys-lajt-ode-1}). This choice of $\phi _{\sigma }$ and $\phi
_{\varepsilon }$ corresponds to the fractional Zener model (\ref{Zener}).
Then $\tilde{\varepsilon}$ and $\tilde{\sigma},$ given by (\ref{eps-sig}),
become
\begin{equation*}
\tilde{\varepsilon}\left( s\right) =F_{0}\frac{s}{s^{2}+\omega ^{2}}\frac{%
\frac{1+as^{\alpha }}{1+bs^{\alpha }}}{1+s^{2}\frac{1+as^{\alpha }}{%
1+bs^{\alpha }}},\;\;\;\;\tilde{\sigma}\left( s\right) =F_{0}\frac{s}{%
s^{2}+\omega ^{2}}\frac{1}{1+s^{2}\frac{1+as^{\alpha }}{1+bs^{\alpha }}}%
,\;\;s\in D.
\end{equation*}%
In the special case $a=b,$ which corresponds to an elastic body, we obtain%
\begin{equation*}
\tilde{\varepsilon}\left( s\right) =F_{0}\frac{s}{s^{2}+\omega ^{2}}\frac{1}{%
1+s^{2}},\;\;\;\;\tilde{\sigma}\left( s\right) =F_{0}\frac{s}{s^{2}+\omega
^{2}}\frac{1}{1+s^{2}}.
\end{equation*}%
After inverting the Laplace transforms, we have
\begin{gather}
\varepsilon \left( t\right) =\frac{F_{0}}{\omega ^{2}-1}\cos \left( \omega
t\right) \ast \sin t,\;\;\;\;\sigma \left( t\right) =\frac{F_{0}}{\omega
^{2}-1}\cos \left( \omega t\right) \ast \sin t,  \label{el-eps-1} \\
\varepsilon \left( t\right) =\frac{2F_{0}}{\omega ^{2}-1}\sin \frac{\left(
\omega +1\right) t}{2}\sin \frac{\left( \omega -1\right) t}{2}%
,\;\;\;\;\sigma \left( t\right) =\frac{2F_{0}}{\omega ^{2}-1}\sin \frac{%
\left( \omega +1\right) t}{2}\sin \frac{\left( \omega -1\right) t}{2}.
\notag  \label{el-eps-2}
\end{gather}%
For $\omega \rightarrow 1$ we obtain a resonance. In this case (\ref%
{el-eps-1}) become%
\begin{equation*}
\varepsilon \left( t\right) =\frac{1}{2}t\sin t,\;\;\;\;\sigma \left(
t\right) =\frac{1}{2}t\sin t,
\end{equation*}%
Also, in the case when $\omega \approx 1$ one observes the pulsation.

We shall present several plots of $\varepsilon ,$ whose explicit form is
obtained in Theorem \ref{thm}, (\ref{u-eps}), in the case of the fractional
Zener model of the viscoelastic body (\ref{zener}). We fix the parameters of
the model: $a=0.2,$ $b=0.6$ and in following figures we present plots of $%
\varepsilon $ in the cases when the forcing term is given as $F=\delta ,$ $%
F=H$ ($H$ is the Heaviside function) and $F\left( t\right) =\cos \left(
\omega t\right) .$

Let $F=\delta .$ We see from Figure \ref{p-f-1}
\begin{figure}[htbp]
\centering
\includegraphics[scale=1]{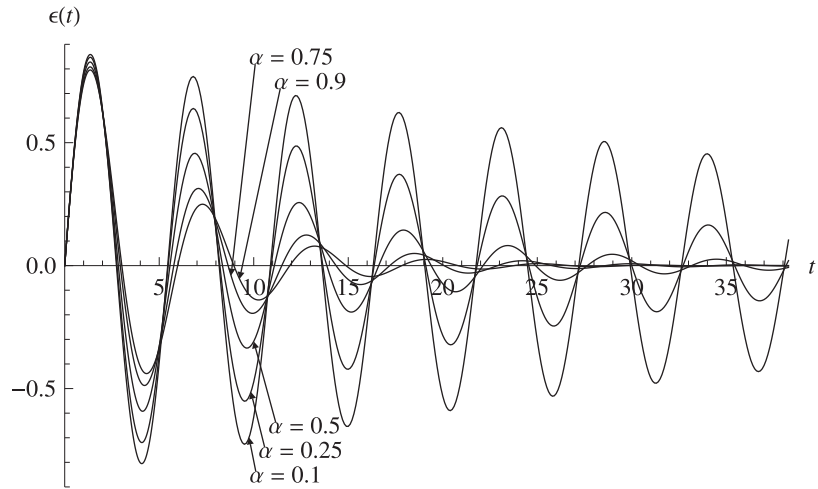}
\caption{Strain $\protect\varepsilon (t)$ in the case $F=\protect\delta $ as
a function of time $t\in (0,38)$.}
\label{p-f-1}
\end{figure}
that the oscillations of the body are damped, since the rod is viscoelastic.
The curve resembles to the curve of the damped oscillations of the linear
harmonic oscillator. We also see that the change of $\alpha \in \left\{
0.1,0.25,0.5,0.75,0.9\right\} $ makes the amplitudes of the curves to
decrease slower with the time, as $\alpha $ becomes smaller. This is due the
the fact that $\alpha =1$ corresponds to the standard linear viscoelastic
body and $\alpha =0$ corresponds to the elastic body. In the case of the
forcing term given as the Heaviside function, from Figures \ref{p-f-2} and %
\ref{p-f-2a}
\begin{figure}[htbp]
\begin{minipage}{72mm}
 \centering
 \includegraphics[scale=0.7]{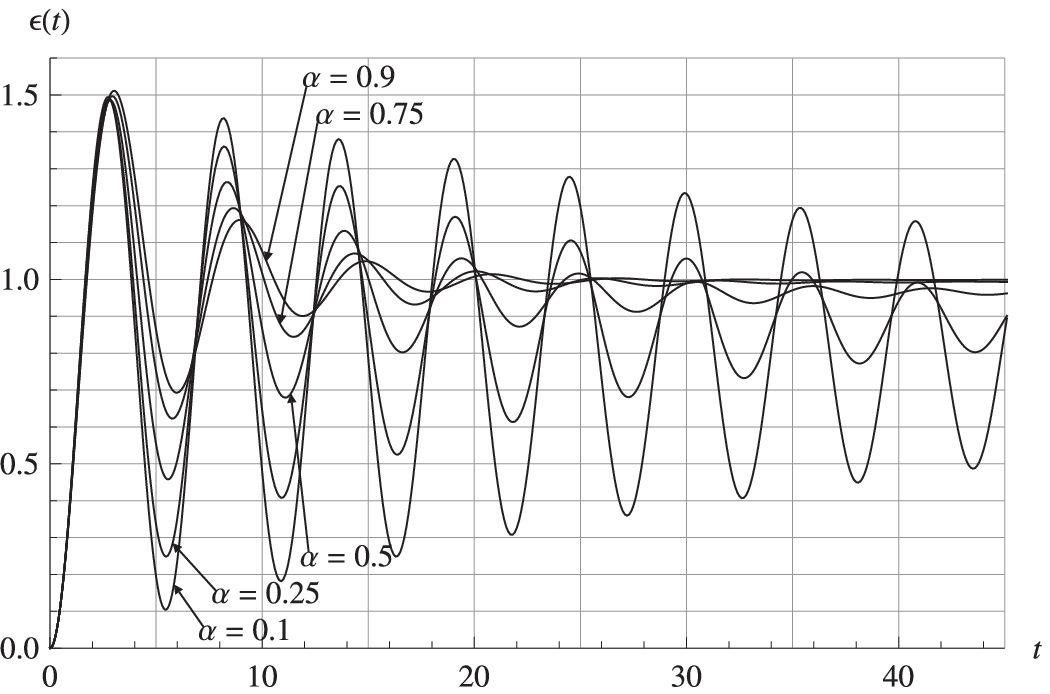}
 \caption{Strain $\protect\varepsilon(t)$ in the case $F=H$ as
a function of time $t\in (0,45)$.}
 \label{p-f-2}
 \end{minipage}
\hfil
\begin{minipage}{72mm}
 \centering
 \includegraphics[scale=0.7]{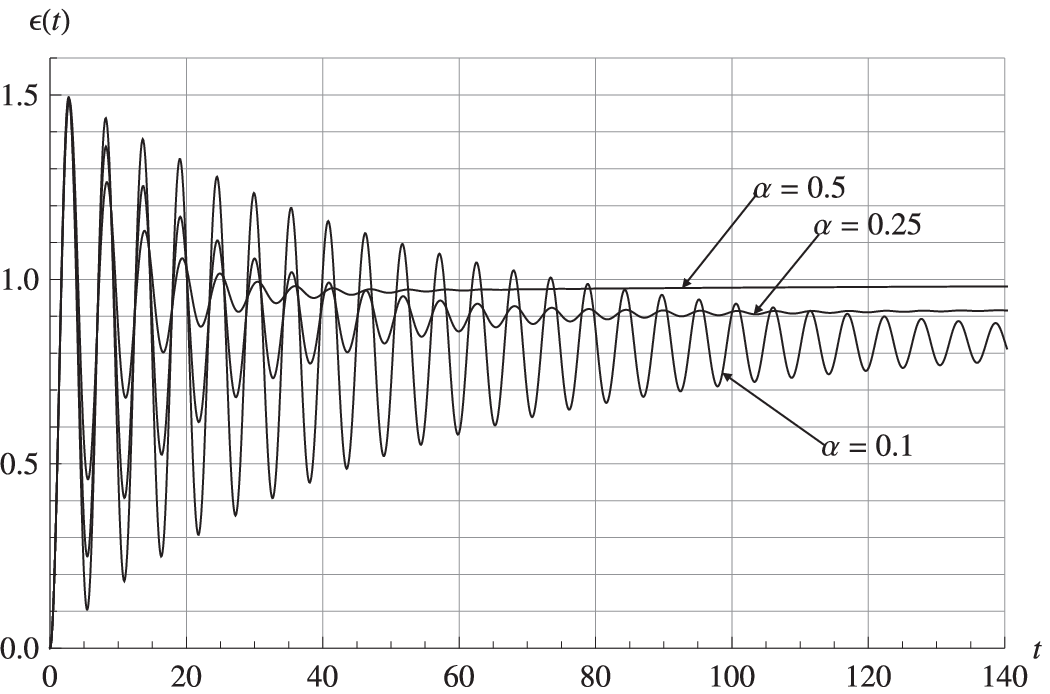}
 \caption{Strain $\protect\varepsilon(t)$ in the case $F=H$ as
a function of time $t\in (0,140)$.}
 \label{p-f-2a}
 \end{minipage}
\end{figure}
we observe that the body creeps to the finite value of the displacement
regardless of the value of $\alpha \in \left\{ 0.1,0.25,0.5,0.75,0.9\right\}
.$ Creeping to the finite value of the displacement is due to the fact that
the fractional Zener model describes the solid-like viscoelastic material.
If the viscoelastic properties of the material are dominant (the value of $%
\alpha $ is closer to one) then the time required for body to reach the
limiting value of strain is smaller, see Figure \ref{p-f-2}, compared to
time in the case when elastic properties of the material are dominant, see
Figure \ref{p-f-2a}. Figure \ref{p-f-3}
\begin{figure}[p]
\centering
\includegraphics[scale=1]{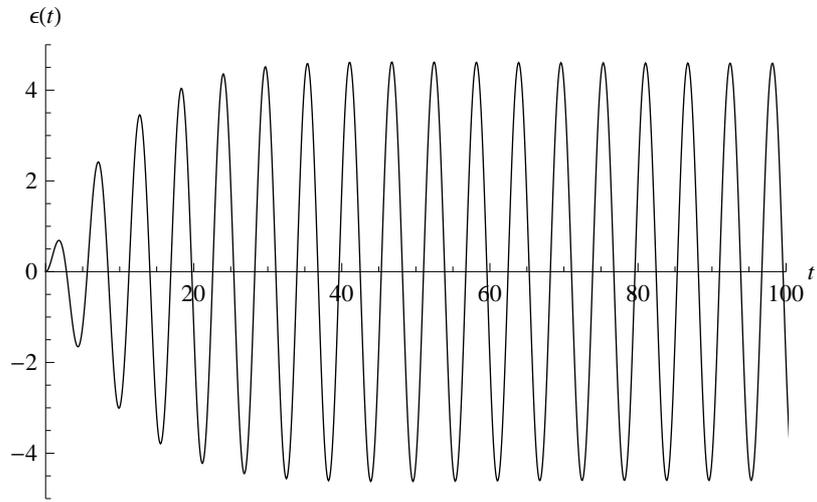}
\caption{Strain $\protect\varepsilon (t)$ in the case $F(t)=\cos (\protect%
\omega t)$ as a function of time $t\in (0,100)$.}
\label{p-f-3}
\end{figure}
shows the expected behavior of the body in the case of the harmonic forcing
term. Namely, the oscillations of the body die out and the body oscillates
in the phase with the harmonic function. For this plot we took: $\alpha =0.45
$ and $\omega =1.1.$

Now, we examine the case when the values of the coefficients $a$ and $b$ are
close to each other. We fix them to be $a=0.58$ and $b=0.6.$ In this case
the elastic properties of the material prevail, since in the limiting case $%
a=b$ the fractional Zener model (\ref{Zener}) becomes the Hooke law for the
arbitrary value of $\alpha \in \left( 0,1\right) .$ We present in Figure \ref%
{p-f-5}
\begin{figure}[p]
\centering
\includegraphics[scale=1]{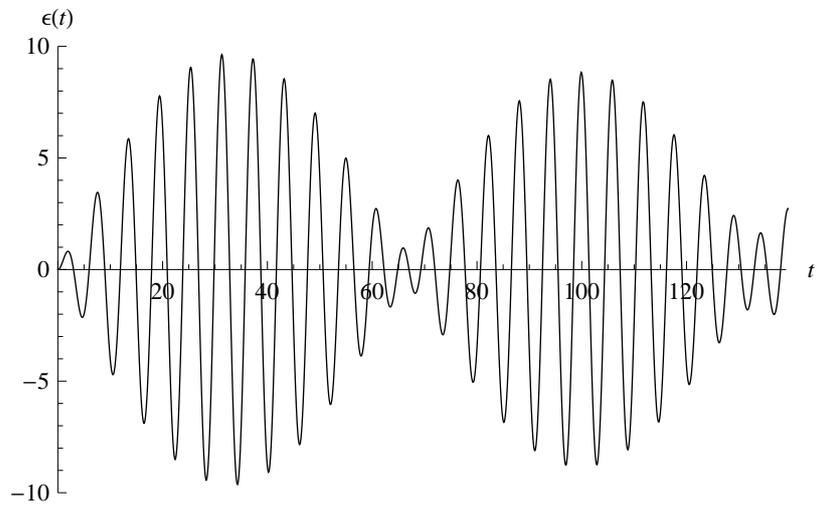}
\caption{Strain $\protect\varepsilon (t)$ in the case $F(t)=\cos (\protect%
\omega t)$ as a function of time $t\in (0,139)$.}
\label{p-f-5}
\end{figure}
the plot of $\varepsilon $ in the case when $\alpha =0.45$ and $\omega =1.1.$
In the elastic case, as it can be seen from (\ref{el-eps-1}), the frequency
of the free oscillations of the body is $\omega _{f}=1.$ Since the frequency
of the forcing function ($\omega =1.1$) is close to the frequency of the
free oscillations, we shall have the pulsation, as it can be observed from
Figure \ref{p-f-5}. Since there is still some damping left ($a\neq b$) the
amplitude of the wave-package decreases in time, as shown in Figure \ref%
{p-f-4}
\begin{figure}[p]
\centering
\includegraphics[scale=1]{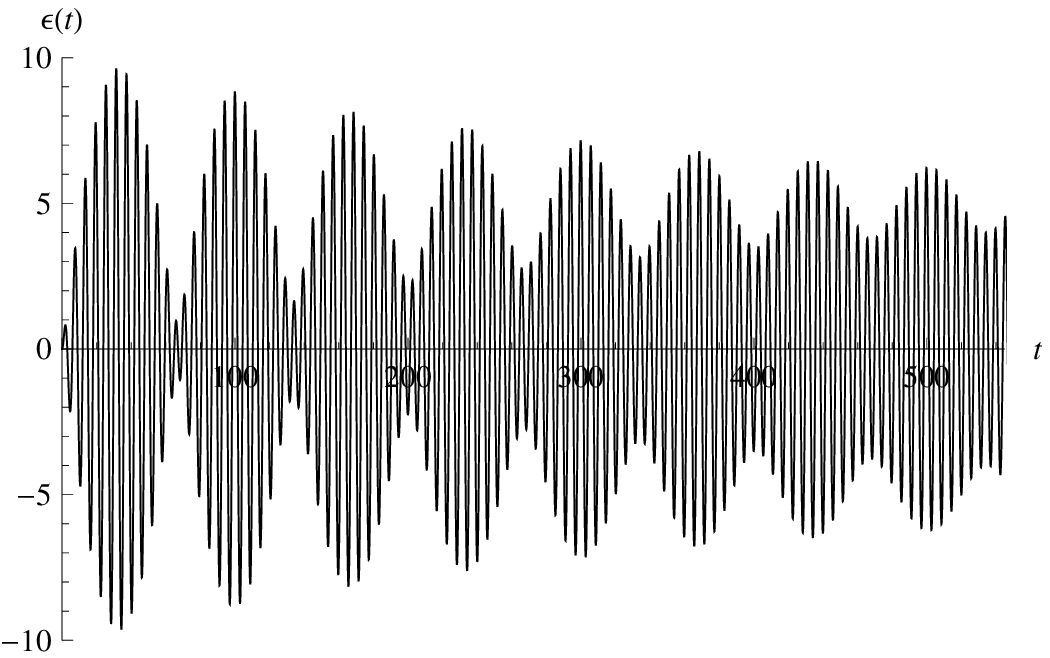}
\caption{Strain $\protect\varepsilon (t)$ in the case $F(t)=\cos (\protect%
\omega t)$ as a function of time $t\in (0,545)$.}
\label{p-f-4}
\end{figure}

We increase the damping effect by choosing $a=0.55,$ $b=0.6.$ The rest of
the parameters are: $\alpha =0.45,$ $\omega =1.1.$ The curve of $\varepsilon
,$ presented in Figure \ref{p-f-6}
\begin{figure}[p]
\centering
\includegraphics[scale=1]{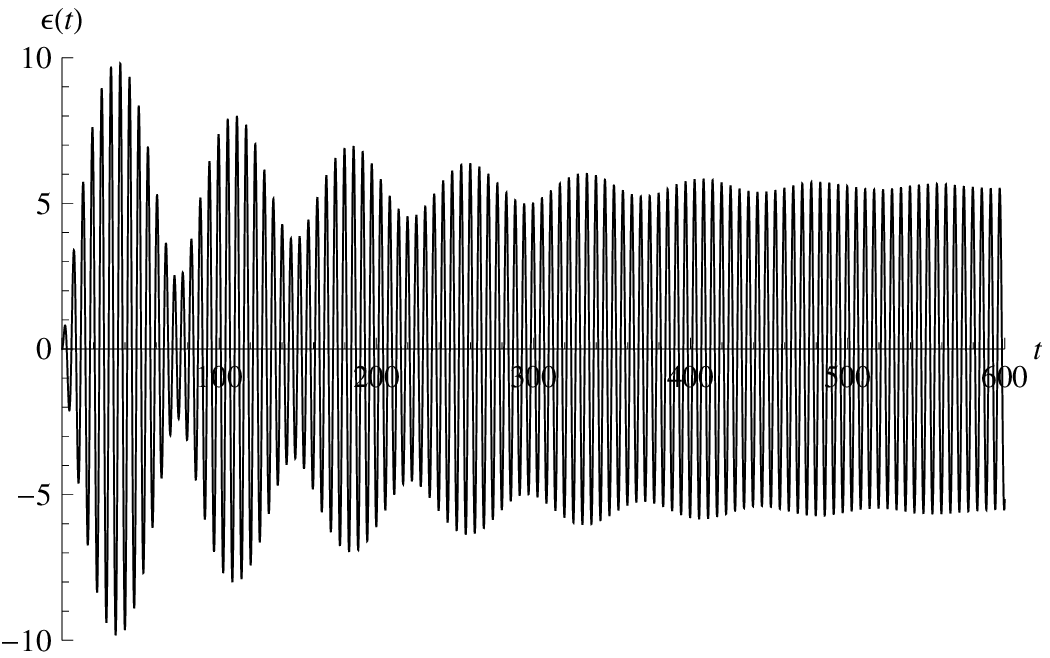}
\caption{Strain $\protect\varepsilon (t)$ in the case $F(t)=\cos (\protect%
\omega t)$ as a function of time $t\in (0,600)$.}
\label{p-f-6}
\end{figure}
for smaller times resembles to the curve of the pulsation (the elastic
properties of the material prevail), while later it resembles to the curve
of the forcing function (the viscous properties of the material prevail).

\begin{acknowledgement}
This research is supported by the Serbian Ministry of Education and Science
projects $174005$ (TMA and DZ) and $174024$ (SP), as well as by the
Secretariat for Science of Vojvodina project $114-451-2167$ (DZ).
\end{acknowledgement}

\end{document}